\documentclass[10pt,conference]{IEEEtran}
\usepackage{graphicx}
\usepackage{amsmath}
\usepackage{amssymb}
\usepackage{amsfonts}

\newcommand{\myincludegraphics}[2]{\includegraphics{#1.#2}}
\newcommand{\fmult}{\ast}

\newcommand{\Markovchain}{\mathcal{P}}

\newcommand{\diam}[1]{\mathop{\mathrm{diam}}(#1)}
\newcommand{\size}[1]{\left|#1\right|}
\newcommand{\set}[1]{\left\{#1\right\}}
\newcommand{\floor}[1]{\left\lfloor#1\right\rfloor}
\newcommand{\ceil}[1]{\left\lceil#1\right\rceil}
\newcommand{\roundb}[1]{\left(#1\right)}
\newcommand{\bracket}[1]{\langle#1\rangle}

\newcommand{\bfone}{\mathbf{1}}
\newcommand{\bfpi}{\boldsymbol{\pi}}
\newcommand{\bfe}{\boldsymbol{e}}

\newcommand{\bfr}{\boldsymbol{r}}
\newcommand{\bfg}{\boldsymbol{g}}
\newcommand{\bfrho}{\mbox{\boldmath$\rho$}}
\newcommand{\bfgamma}{\mbox{\boldmath$\gamma$}}
\newcommand{\bflambda}{\mbox{\boldmath$\lambda$}}
\newcommand{\bfralpha}{\tilde{\bfrho}}
\newcommand{\bfrbeta}{\tilde{\bfr}}
\newcommand{\ralpha}{\tilde{\rho}}
\newcommand{\rbeta}{\tilde{r}}
\newcommand{\Dalpha}{\tilde{P}}
\newcommand{\dalpha}{\tilde{p}}
\newcommand{\Da}{D}
\newcommand{\da}{d}

\newcommand{\ra}{r}

\newcommand{\Malpha}{M'}
\newcommand{\Surplus}{\mathrm{Surp}}

\newcommand{\bfv}{\boldsymbol{v}}
\newcommand{\bfx}{\boldsymbol{x}}
\newcommand{\bfxi}{\boldsymbol{\xi}}
\newcommand{\sfc}{\mathsf{c}}
\newcommand{\sfe}{\mathsf{e}}
\newcommand{\sfG}{\mathsf{G}}
\newcommand{\sfN}{\mathsf{N}}

\newcommand{\rfp}{\rho}
\newcommand{\bfrfp}{\bfrho}
\newcommand{\Phifp}{P}
\newcommand{\phifp}{p}
\newcommand{\capacity}{\mathsf{cap}}
\newcommand{\mmax}[1]{\mathop{\max\set{#1}}}
\newcommand{\mmin}[1]{\mathop{\min\set{#1}}}

\newcommand{\wt}{{w_\mathrm{d}}} 
\newcommand{\wm}{{w_\mathrm{m}}} 
\newcommand{\GK}{G^{\otimes N}}
\newcommand{\GKK}{G^{\otimes \sfN}}
\newcommand{\sfGKK}{\sfG^{\otimes \sfN}}
\newcommand{\usource}{u_\sigma}
\newcommand{\utarget}{u_\tau}
\newcommand{\amin}{a_{\min}}
\newcommand{\amax}{a_{\max}}
\newcommand{\memory}{\mathsf{m}}
\newcommand{\anticipation}{\mathsf{a}}
\newcommand{\parent}{\mathrm{parent}}

\newcommand{\Equation}[1]{\hbox{(\ref{#1})}}

\newcommand{\Footnotetext}[2]{\begin{figure}[b]\footnotesize%
  \vspace{-3ex}\hrulefill\hfill\makebox[0em]{}\hfill\makebox[0em]{}%
  \vfill
  \par${}^{#1}$ #2\vspace{-0.60ex}\end{figure}\addtocounter{figure}{0}}

\newtheorem{theo}{Theorem}
\newtheorem{lemm}[theo]{Lemma}

\newcommand{\fbinom}[2]{\genfrac{\langle}{\rangle}{0pt}{}{#1}{#2}}
\newcommand{\ifbinom}[2]{\genfrac{\lceil}{\rceil}{0pt}{}{#1}{#2}}
\newcommand{\ceiling}[1]{\left\lceil #1 \right\rceil}
\newcommand{\float}[1]{\overline{#1}}
\newcommand{\calS}{\mathbb{S}}

\newenvironment{algorithm}{%
\begin{minipage}{\columnwidth}\vspace{1ex}\small
\makebox[0ex]{}\hrulefill\makebox[0ex]{}\\*}{%
\makebox[0ex]{}\hrulefill\makebox[0ex]{}\end{minipage}}

\newenvironment{example}[1][Example:]{%
\noindent\textbf{#1}
}{%
\hspace{10pt}\hfill$\blacksquare$\par
}

\begin{document}

\title{On Row-by-Row Coding for 2-D Constraints
}

\author{\authorblockN{Ido Tal \qquad Tuvi Etzion \qquad Ron M. Roth}
\authorblockA{
Computer Science Department,\\
Technion, Haifa 32000, Israel.\\
Email: \{{\tt idotal, etzion, ronny}\}{\tt @cs.technion.ac.il}
}}
\maketitle

\begin{abstract}
A constant-rate encoder--decoder pair is presented
for a fairly large family of two-dimensional (2-D) constraints.
Encoding and decoding is done in a row-by-row manner,
and is sliding-block decodable.

Essentially, the 2-D constraint is turned into a set of independent
and relatively simple one-dimensional (1-D) constraints;
this is done by dividing the array into fixed-width vertical strips.
Each row in the strip is seen as a symbol, and a graph presentation
of the respective 1-D constraint is constructed. The maxentropic
stationary Markov chain on this graph is next considered: a perturbed
version of the corresponding probability distribution on the edges of
the graph is used in order to build an encoder which operates
\emph{in parallel} on the strips. This perturbation is found by means
of a network flow, with upper and lower bounds on the flow through
the edges.

A key part of the encoder is an enumerative coder for constant-weight
binary words. A fast realization of this coder is shown, using
floating-point arithmetic.
\end{abstract}
\Footnotetext{}
{
The work of Tuvi Etzion was supported in part by the United States --
Israel Binational Science Foundation (BSF), Jerusalem, Israel,
under Grant No. 2006097.

The work of Ron M. Roth was supported in part by the United States --
Israel Binational Science Foundation (BSF), Jerusalem, Israel,
under Grant No. 2002197.}

\section{Introduction}
\label{sec:introduction}
Let $G = (V,E,L)$ be an edge-labeled directed graph
(referred to hereafter simply as a graph),
where $V$ is the vertex set, $E$ is the edge set,
and $L:E \to \Sigma$ is the edge labeling taking values on
a finite alphabet $\Sigma$~\cite[\S 2.1]{MarcusRothSiegel:98}.
We require that the labeling $L$ is deterministic: edges that start
at the same vertex have distinct labels. We further assume
that $G$ has finite memory~\cite[\S 2.2.3]{MarcusRothSiegel:98}.
The one-dimensional (1-D) \emph{constraint} $S=S(G)$ that is presented
by $G$ is defined as the set of all words that are generated by paths
in $G$ (i.e., the words are obtained by reading-off the edge labels of
such paths).
Examples of 1-D constraints include
runlength-limited (RLL)
constraints~\cite[\S 1.1.1]{MarcusRothSiegel:98},
symmetric runlength-limited (SRLL) constraints~\cite{Etzion:97},
and the charge constraints~\cite[\S 1.1.2]{MarcusRothSiegel:98}.
The capacity of $S$ is given by
\[
\capacity(S) =
\lim_{\ell \to \infty}
(1/\ell) \cdot \log_2 \size{S \cap \Sigma^\ell} \; .
\]

An $M$-track \emph{parallel encoder} for $S=S(G)$
at rate $R$ is defined as follows (see Figure~\ref{fig:encoder}).

\begin{figure}[ht]
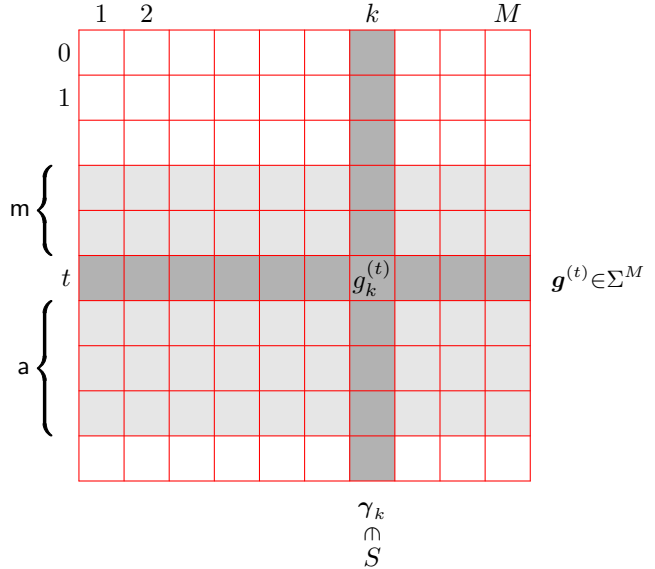

\centering
\myincludegraphics{flow}{1}
\caption{Array corresponding to an $M$-track parallel encoder.}
\label{fig:encoder}
\end{figure}

\begin{enumerate}
\item
At stage $t = 0, 1, 2, \cdots$, the encoder
(which may be state-dependent) receives as
input $M \cdot R$ (unconstrained) information bits.
\item
The output of the encoder at stage $t$ is
a word $\bfg^{(t)} = (g_k^{(t)})_{k=1}^M$
of length $M$ over $\Sigma$.
\item
For $1 \leq k \leq M$, the \emph{$k$th track}
$\bfgamma_k = (g_k^{(t)})_{t=0}^{\ell-1}$
of any given length $\ell$, belongs to $S$.
\item
There are integers $\memory,\anticipation \geq 0$
such that the encoder is
\emph{$(\memory,\anticipation)$-sliding-block decodable}
(in short, $(\memory,\anticipation)$-SBD):
for $t \geq \memory$, the $M \cdot R$ information bits which were
input at stage $t$ are uniquely determined by
(and can be efficiently calculated from)
$\bfg^{(t-\memory)}, \bfg^{(t-\memory+1)}, \ldots,
\bfg^{(t+\anticipation)}$.
\end{enumerate}
The decoding window size of the encoder is
$\memory + \anticipation + 1$, and it is desirable to have a small
window to avoid error propagation. In this work, we will be mainly
focusing on the case where $\anticipation = 0$, in which case
the decoding requires no look-ahead.

In~\cite{HalevyRoth:02}, it was shown that by introducing parallelism,
one can reduce the window size, compared to conventional serial
encoding. Furthermore, it was shown that as $M$ tends to infinity,
there are $(0,0)$-SBD parallel encoders whose rates
approach $\capacity(S(G))$.
A key step in~\cite{HalevyRoth:02} is using some perturbation of
the conditional probability distribution on the edges of $G$,
corresponding to the maxentropic stationary Markov chain on $G$.
However, it is not clear how this perturbation should be done:
a naive method will only work for unrealistically large $M$.
Also, the proof in~\cite{HalevyRoth:02} of the $(0,0)$-SBD property
is only probabilistic and does not suggest encoders and decoders
that have an acceptable running time.

In this work, we aim at making the results
of~\cite{HalevyRoth:02} more tractable. At the expense of
possibly increasing the memory of the encoder (up to the memory of $G$)
we are able to define a suitable perturbed distribution explicitly,
and provide an efficient algorithm for computing it. Furthermore,
the encoding and decoding can be carried out in time complexity
$O(M \log^2 M \log \log M)$,
where the multiplying constants in the $O(\cdot)$
term are polynomially large in the parameters of $G$.

Denote by $\diam{G}$ the diameter of $G$ (i.e., the longest shortest
path between any two vertices in $G$) and let $A_G = (a_{i,j})$
be the adjacency matrix of $G$,
i.e., $a_{i,j}$ is the number of edges in $G$ that start at vertex $i$
and terminate in vertex $j$.
Our main result, specifying the rate of our encoder, is given in
the next theorem.

\begin{theo}
\label{theo:main}
Let $G$ be a deterministic graph with memory $\memory$.
For $M$ sufficiently large, one can efficiently construct
an $M$-track $(\memory,0)$-SBD parallel encoder for $S = S(G)$
at a rate $R$ such that
\begin{multline}
\label{eq:Rbound}
R \geq \capacity(S(G)) \Bigl( {1 -\frac{\size{V} \diam{G}}{2M}} \Bigr)
\\
- O\roundb{\frac{\size{V}^2
\log \, (M \cdot \amax/\amin)}{M - \size{V} \diam{G}/2}} \; ,
\end{multline}
where $\amin$ (respectively, $\amax$) is the smallest
(respectively, largest) \emph{nonzero} entry in $A_G$.
\end{theo}

The structure of this paper is as follows. In Section~\ref{sec:twoDimensionalConstraints} we show how parallel encoding can be used to construct an encoder for a 2-D constraint. As we will show, a parallel encoder is essentially defined through what we term a multiplicity matrix. Section~\ref{sec:encoder} defines how our parallel encoder works, assuming its multiplicity matrix is given. Then, in Section~\ref{sec:computingmultiplicitymatrix}, we show how to efficiently calculate a good multiplicity matrix. Although 2-D constraints are our main motivation, Section \ref{sec:1D} shows how our method can be applied to 1-D constraints. Section~\ref{sec:twoOpt} defines two methods by which the rate of our encoder can be slightly improved. Finally, in  Section~\ref{sec:fastenumerativecoding} we show a method of efficiently realizing a key part of our encoding procedure.
\section{Two-dimensional constraints}
\label{sec:twoDimensionalConstraints}

Our primary motivation for studying parallel encoding is to show
an encoding algorithm for a family of two-dimensional (2-D) constraints.

The concept of a 1-D constraint can formally be generalized
to two dimensions (see~\cite[\S 1]{HalevyRoth:02}).
Examples of 2-D constraints are
2-D RLL constraints~\cite{KatoZeger:99},
2-D SRLL constraints~\cite{Etzion:97}, and
the so-called square constraint~\cite{WeeksBlahut:98}.
Let $\calS$ be a given 2-D constraint over a finite alphabet $\Sigma$.
We denote by $\calS[\ell,w]$ the set of all $\ell \times w$ arrays
in $\calS$.
The capacity of $\calS$~\cite{BurtonSteif:94}
is given by
\[
\capacity(\calS) =
\lim_{\ell, w \to \infty} \frac{1}{\ell \cdot w} \cdot
\log_2 \size{\calS[\ell, w]} \; .
\]

Suppose we wish to encode information to an $\ell \times w$ array which
must satisfy the constraint $\calS$; namely, the array must be
an element of $\calS[\ell,w]$. As a concrete example, consider
the square constraint~\cite{WeeksBlahut:98}:
its elements are all the binary arrays in which no two `1' symbols are adjacent on a row, column, or diagonal.

We first partition our array into two alternating types of vertical
strips: \emph{data strips} having width $\wt$, and
\emph{merging strips} having width $\wm$.
In our example, let $\wt=4$ and $\wm=1$
(see Figure~\ref{fig:example:square}).

\begin{figure}[htbf]
\small
\[
\begin{array}{cccc|c|cccc|c|cccc}
0 & 0 & 1 & 0 & 0 & 1 & 0 & 1 & 0 & 0 & 0 & 0 & 0 & 1 \\
1 & 0 & 0 & 0 & 0 & 0 & 0 & 0 & 0 & 0 & 1 & 0 & 0 & 0 \\
0 & 0 & 0 & 1 & 0 & 0 & 1 & 0 & 0 & 0 & 0 & 0 & 0 & 0 \\
1 & 0 & 0 & 0 & 0 & 0 & 0 & 0 & 1 & 0 & 1 & 0 & 0 & 1 \\
\end{array}
\]
\caption{Binary array satisfying the square constraint,
partitioned into data strips of width $\wt=4$
and merging strips of width $\wm=1$.}
\label{fig:example:square}
\end{figure}

Secondly, we select a graph $G = (V,E,L)$ with
a labeling $L:E \to \calS[1, \wt]$ such that
$S(G) \subseteq \calS$, i.e.,
each path of length $\ell$ in $G$ generates
a (column) word which is in $\calS[\ell,\wt]$.
We then fill up the data strips of our $\ell \times w$ array
with $\ell \times \wt$ arrays corresponding to paths
of length $\ell$ in $G$.
Thirdly, we assume that the choice of $\wm$ allows us
to fill up the merging strips in a row-by-row (causal)
manner, such that our $\ell \times w$ array is in $\calS$.
Any 2-D constraint $\calS$ for which such $\wt$, $\wm$,
and $G$ can be found, is in the family of constraints we can code for
(for example, the 2-D SRLL constraints
belong to this family~\cite{Etzion:97}).

Consider again the square constraint: a graph which produces
\emph{all} $\ell \times \wt$ arrays that satisfy this constraint
is given in Figure~\ref{fig:example:G}.
Also, for $\wm = 1$, we can take the merging strips to be all-zero.
(There are cases, such as the 2-D SRLL constraints,
where determining the merging strips may be
less trivial~\cite{Etzion:97}.)
\begin{figure}[ht]
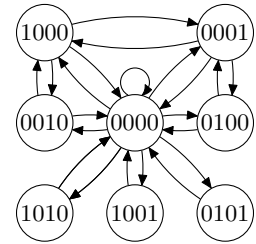

\centering
\myincludegraphics{flow}{2}
\caption{Graph $G$ whose paths generate all $\ell \times 4$
arrays satisfying the square constraint.
The label of an edge is given by the label of the vertex
it enters.}
\label{fig:example:G}
\end{figure}

Suppose we have
an $(\memory,0)$-SBD parallel encoder for $S = S(G)$ at rate $R$
with $M = (w + \wm)/(\wt+\wm)$ tracks.
We may use this parallel encoder to encode information in
a row-by-row fashion to our $\ell \times w$ array: at stage $t$
we feed $M \cdot R$ information bits to our parallel encoder.
Let $\bfg^{(t)} = (g_k^{(t)})_{k=1}^M$ be the output of
the parallel encoder at stage $t$. We write $g_k^{(t)}$
to row $t$ of the $k$th data strip, and then appropriately fill up
row $t$ of the merging strips.
Decoding of a row in our array can be carried out based only on
the contents of that row and the previous $\memory$ rows.

Since $M \cdot R$ information bits are mapped to $M \cdot \wt + (M-1)\cdot \wm$ symbols in $\Sigma$, the rate at which we encode information to the array is
\[
\frac{R}{\wt + \wm(1 - 1/M)} \leq
\frac{\capacity(S(G))}{\wt + \wm(1 - 1/M)} \; .
\]
We note the following tradeoff: Typically, taking larger values of $\wt$ (while keeping $\wm$ constant) will increase the right-hand side of the above inequality.
However, the number of vertices and edges in $G$ will usually grow exponentially
with $\wt$. Thus, $\wt$ is taken to be reasonably small. 

Note that in our scheme, a single error generally results in the loss of information stored in the respective vertical sliding-block window. Namely, a single corrupted entry in the array may cause the loss of $\memory+1$ rows. Thus, our method is only practical if we assume an error model in which whole rows are corrupted by errors. This is indeed the case if each row is protected by an error-correcting code (for example, by the use of unconstrained positions \cite{Campello+:02}).

\section{Description of the encoder}
\label{sec:encoder}
Let $N$ be a positive integer which will shortly be specified.  The $N$ words $\bfgamma_k = (g_k^{(t)})_{t=0}^{\ell-1}$,
$1 \leq k \leq N$, that we will be writing to the first $N$ tracks are
all generated by
paths of length $\ell$ in $G$.
In what follows, we find it convenient to regard
the $\ell \times N$ arrays
$(\bfgamma_k)_{k=1}^N = (g_k^{(t)})_{t=1}^\ell{}_{k=1}^N$
as (column) words of length $\ell$ of some new 1-D constraint,
which we define next.

The \emph{$N$th Kronecker power} of $G = (V,E,L)$, denoted by $\GK = (V^N,E^N,L^N)$,  is defined as follows.
The vertex set $V^N$ is simply the $N$th Cartesian power of $V$;
that is,
\[
V^N = \set{ \bracket{v_1, v_2, \ldots, v_N} : v_k \in V } \; .
\]
An edge $\bfe = \bracket{e_1, e_2, \ldots, e_N} \in E^N$ goes from vertex
$\bfv =  \bracket{v_1, v_2, \ldots, v_N} \in V^N$ to vertex
$\bfv' = \bracket{v_1', v_2', \ldots, v_N'} \in V^N$ and is labeled
$L^N(\bfe) = \bracket{b_1, b_2, \ldots, b_N}$ whenever for
all $1 \leq k \leq N$, $e_k$ is
an edge from $v_k$ to $v_k'$ labeled $b_k$.

Note that a path of length $\ell$ in $\GK$ is just
a handy way to denote $N$ paths of length $\ell$ in $G$.
Accordingly, the $\ell \times N$ arrays
$(\bfgamma_k)_{k=1}^N$ are the words of length $\ell$ in $S(\GK)$.

Let $G$ be as in Section~\ref{sec:introduction} and
let $A_G = (a_{i,j})$ be the adjacency matrix of $G$.
Denote by $\bfone$ the $1 \times \size{V}$ all-one row vector.
The description of
our $M$-track  parallel encoder for $S=S(G)$ makes use
of the following definition.
A $\size{V} \times \size{V}$ nonnegative integer matrix
$\Da=(\da_{i,j})_{i,j \in V}$ is called a (valid)
\emph{multiplicity matrix} with respect to $G$ and $M$ if
\begin{align}
\label{eq:sum_entries}
& \bfone \cdot \Da \cdot \bfone^T \leq M \; ,\\
\label{eq:sum_up}
& \bfone \cdot \Da = \bfone \cdot \Da^T \; , \quad \textrm{and} \\
\label{eq:sane_entries}
& \textrm{$\da_{i,j} > 0$ only if $a_{i,j} > 0$} \; .
\end{align}
(While any multiplicity matrix will produce a parallel encoder,
some will have higher rates than others.
In Section~\ref{sec:computingmultiplicitymatrix},
we show how to compute multiplicity matrices $\Da$ that yield rates
close to $\capacity(S(G))$.)

Recall that we have at our disposal $M$ tracks. However,
we will effectively be using only the first
$N = \bfone \cdot \Da \cdot \bfone^T$ tracks in order
to encode information.
The last $M-N$ tracks  will all be equal to the first track, say.

Write $\bfr = (r_i)_{i \in V} = \bfone \cdot \Da^T$.
A vertex $\bfv = \bracket{v_k}_{k=1}^N \in V^N$ is
a \emph{typical vertex} (with respect to $\Da$) if for
all $i$, the vertex $i \in V$ appears as an entry in $\bfv$ exactly
$r_i$ times.
Also, an edge $\bfe = \bracket{e_k}_{k=1}^N \in E^N$ is
a \emph{typical edge} with respect to $\Da$ if for all $i, j \in V$,
there are exactly $\da_{i,j}$ entries $e_k$
which---as edges in $G$---start at vertex $i$
and terminate in vertex $j$.

A simple computation shows that the number of
outgoing typical edges from a typical vertex equals
\begin{equation}
\label{eq:R_D}
\Delta =
\frac{\prod_{i \in V} r_i!}{\prod_{i,j \in V} \da_{i,j}!
\cdot a_{i,j}^{-\da_{i,j}}}
\end{equation}
(where $0^0 \triangleq 1$).
For example, in the simpler case where $G$ does not contain parallel
edges ($a_{i,j} \in \set{0,1}$),
we are in effect counting in~\Equation{eq:R_D} permutations with
repetitions, each time for a different vertex $i \in V$.

The encoding process will be carried out as follows. We start at some
fixed typical vertex $\bfv^{(0)} \in V^N$.
Out of the set of outgoing edges from $\bfv^{(0)}$,
we consider only typical edges.
The edge we choose to traverse is determined by
the information bits. After traversing the chosen edge,
we arrive at vertex $\bfv^{(1)}$.
By~\Equation{eq:sum_up}, $\bfv^{(1)}$ is also a typical vertex,
and the process starts over.
This process defines
an $M$-track parallel encoder for $S=S(G)$ at rate
\begin{equation*}
R = R(\Da) = \frac{\lfloor \log_2 \Delta \rfloor}{M} \; .
\end{equation*}
This encoder is $(\memory,0)$-SBD, where $\memory$ is the memory of $G$.

Consider now how we map $M \cdot R$ information bits into
an edge choice $\bfe \in E^N$ at any given stage $t$.
Assuming again the simpler case of a graph with no parallel edges,
a natural choice would be to use an instance of
enumerative coding~\cite{Cover:73}.
Specifically, suppose that for $0 \leq \delta \leq n$,
a procedure for encoding information by an $n$-bit binary vector
with Hamming weight $\delta$ were given.
Suppose also that $V = \set{1,2,\ldots,\size{V}}$.
We could use this procedure as follows.
First, for $n = r_1$ and $\delta = \da_{1,1}$,
the binary word given as output by the procedure will define
which $\da_{1,1}$ of the possible $r_1$ entries in $\bfe$
will be equal to the edge in $E$ from
the vertex $1 \in V$ to itself
(if no such edge exists, then $\da_{1,1}=0$).
Having chosen these entries,
we run the procedure with $n = r_1 - \da_{1,1}$ and
$\delta = \da_{1,2}$
to choose from the remaining $r_1 - \da_{1,1}$ entries those
that will contain the edge in $E$ from $1 \in V$ to $2 \in V$.
We continue this process, until all $r_1$ entries in $\bfe$ containing
an edge outgoing from $1 \in V$ have been picked.
Next, we run the procedure with $n = r_2$ and
$\delta = \da_{2,1}$, and so forth.
The more general case of a graph containing parallel edges will include
a preliminary step:
encoding information in the choice of
the $\da_{i,j}$ edges used to traverse from $i$ to $j$
($a_{i,j}$ options for each such edge).

A fast implementation of enumerative coding is presented in
Section~\ref{sec:fastenumerativecoding}.
The above-mentioned preliminary step makes use of
the Sch\"onhage--Strassen integer-multiplication
algorithm~\cite[\S 7.5]{Aho+:74},
and the resulting encoding time complexity is proportional\footnote{Actually, the time complexity for the preliminary step can be made linear in $M$, with a negligible penalty in terms of rate: Fix $i$ and $j$, and let $\eta$ be an integer design parameter. Assume for simplicity that $\eta | d_{i,j}$. The number of vectors of length $\eta$ over an alphabet of size $a_{i,j}$ is obviously $a_{i,j}^\eta$. So, we can encode $\floor{\eta \log_2 a_{i,j}}$ bits through the choice of such a vector. Repeating this process, we can encode $(d_{i,j}/\eta) \cdot \floor{\eta \log_2 a_{i,j}}$ bits through the choice of $d_{i,j}/\eta$ such vectors. The concatenation of these vectors is taken to represent our choice of edges. Note that the encoding process is linear in $M$ for constant $\eta$. Also, our losses (due to the floor function) become negligible for modestly large $\eta$.
}
to $M \log^2 M \log\log M$. It turns out that this is also
the decoding time complexity. Further details are given in Section~\ref{sec:fastenumerativecoding}.

The next section shows how to find a good multiplicity matrix,
i.e., a matrix $\Da$ such that $R(\Da)$ is close to $\capacity(S(G))$.

\section{Computing a good multiplicity matrix}
\label{sec:computingmultiplicitymatrix}
In order to enhance the exposition of this section, we accompany it by a running
example (see Figure~\ref{fig:running1}).

\begin{figure}[ht]
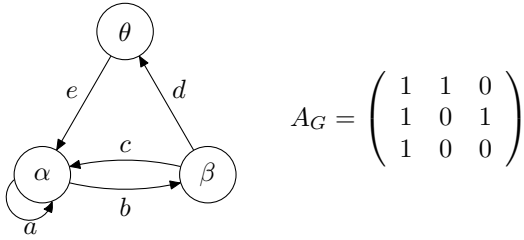

\centering
\parbox{3.5cm}{\myincludegraphics{flow}{3}}
\parbox{3.5cm}{
\[
A_G = \left(
\begin{array}{ccc}
1 & 1 & 0 \\
1 & 0 & 1 \\
1 & 0 & 0
\end{array}
\right) 
\]
}
\caption{\textbf{Running Example (1):} Graph $G$ and the corresponding adjacency matrix
$A_G$.}%
\label{fig:running1}%
\end{figure}%

Throughout this section, we assume a probability distribution on
the edges of $G$, which is
the maxentropic stationary Markov chain
$\Markovchain$ on $G$~\cite{MarcusRothSiegel:98}.
Without real loss of generality,
we can assume that $G$ is irreducible (i.e., strongly-connected),
in which case $\Markovchain$ is indeed unique.
Let the matrix $Q=(q_{i,j})$ be the transition matrix induced by
$\Markovchain$, i.e., $q_{i,j}$ is the probability of traversing
an edge from $i \in V$ to $j \in V$,
conditioned on currently being at vertex $i \in V$.

Let $\bfpi=(\pi_i)$ be the $1 \times \size{V}$ row vector
corresponding to the stationary distribution on $V$ induced by $Q$;
namely, $\bfpi Q = \bfpi$ and $\sum_{i \in V} \pi_i = 1$. Let
\begin{equation}
\label{eq:Malpha}
\Malpha = M - \floor{ \size{V} \diam{G}/2 } \; ,
\end{equation}
and define
\[
\bfrfp = (\rfp_i) \; , \;
\rfp_i = \Malpha \pi_i \; , \quad \textrm{and} \quad
\Phifp = (\phifp_{i,j}) \; ,  \; \phifp_{i,j} = \rfp_i q_{i,j}
\]

\begin{example}[Running Example (2):] Taking the number of tracks in our running example (Figure~\ref{fig:running1}) to be $M=12$  gives $M'=9$. Also, our running example has
\[
\bfpi = \left(
\begin{array}{ccc}
0.619  & 0.282 & 0.099
\end{array}
\right) \; ,
\]
and
\[
Q = \left(
\begin{array}{ccc}
0.544 & 0.456 & 0 \\
0.647 & 0 & 0.353 \\
1 & 0 & 0
\end{array}
\right) \; .
\]
Thus,
\[
\bfrfp = \left(
\begin{array}{ccc}
5.57 &  2.54 & 0.89
\end{array}
\right)
\]
and
\[
\Phifp = \left(
\begin{array}{ccc}
3.03 & 2.54 & 0 \\
1.65 & 0 & 0.89 \\
0.89 & 0 & 0
\end{array}
\right) \; .
\]
\end{example}
Note that 
\[
\bfrfp = \bfone \cdot \Phifp^T \quad \mbox{and} \quad \Malpha = \bfone \cdot
\Phifp \cdot \bfone^T \; .
\]
Also, observe that~\Equation{eq:sum_entries}--\Equation{eq:sane_entries}
hold when we substitute $\Phifp$ for $\Da$. Thus, if all entries of $\Phifp$
were integers, then we could take $\Da$ equal to $\Phifp$.
In a way, that would be the best choice we could have
made: by using Stirling's approximation, we could deduce that
$R(\Da)$ approaches $\capacity(S(G))$ as $M \rightarrow \infty$.
However, the entries of $\Phifp$, as well as $\bfrfp$, may be
non-integers.

%

We say that an \emph{integer} matrix $\Dalpha=(\dalpha_{i,j})$ is
a \emph{good quantization} of $\Phifp=(\phifp_{i,j})$ if
\newcommand{\Align}[1]{\makebox[3.4em]{$#1$}}
\begin{eqnarray}
\label{eq:preflow_sumij}
\textstyle
\Malpha =\sum_{i,j \in V} \phifp_{i,j} & = & \textstyle \sum_{i,j \in V}
\dalpha_{i,j} \; ,\\
\label{eq:preflow_sumj}
\textstyle
\floor{\sum_{j \in V} \phifp_{i,j}}
& \leq & \textstyle \Align{\sum_{j \in V} \dalpha_{i,j}}
\;\;\leq\;\; \ceil{\sum_{j \in V} \phifp_{i,j}} \; ,\\
\label{eq:preflow_entries}
\floor{\phifp_{i,j}} & \leq &
\textstyle \Align{\dalpha_{i,j}}
\;\;\leq\;\; \ceil{\phifp_{i,j}} \; , \quad \textrm{and---}\qquad \\
\label{eq:preflow_sumi}
\textstyle
\floor{\sum_{i \in V} \phifp_{i,j}}
& \leq & \textstyle \Align{\sum_{i \in V} \dalpha_{i,j}}
\;\;\leq\;\; \ceil{\sum_{i \in V} \phifp_{i,j}} \; .
\end{eqnarray}
Namely, a given entry in $\Dalpha$ is either the floor or the ceiling of
the corresponding entry in $\Phifp$, and this also holds for the sum of entries of a
given row or column in $\Dalpha$; moreover, the sum of entries in both
$\Dalpha$ and $\Phifp$ are exactly equal (to $\Malpha$).

\begin{lemm}
\label{lemm:preflow}
There exists a matrix $\Dalpha$ which is a good quantization of
$\Phifp$.
Furthermore, such a matrix can be found by an efficient algorithm.
\end{lemm}

\begin{figure}[ht]
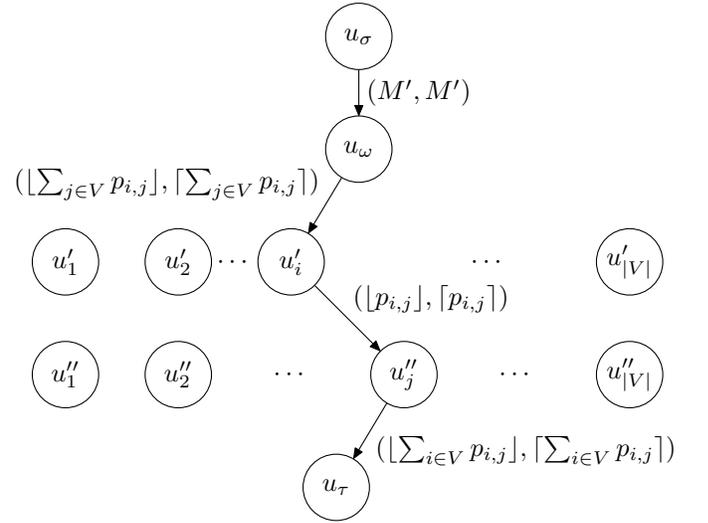

\centering
\myincludegraphics{flow}{4}
\caption{Flow network for the proof of Lemma~\ref{lemm:preflow}. An edge
labeled $(a,b)$ has lower and upper bounds $a$ and $b$, respectively.}
\label{fig:preflow}
\end{figure}

\begin{proof}
We recast~\Equation{eq:preflow_sumij}--\Equation{eq:preflow_sumi} as
an integer flow problem (see Figures~\ref{fig:preflow} and~\ref{fig:running3}).
Consider the following flow network, with upper and lower bounds on
the flow through the edges~\cite[\S 6.7]{Ahuja+:93}.
The network has the vertex set
\[
\set{\usource} \cup \set{u_\omega} \cup \set{\utarget} \cup
\set{u_i'}_{i \in V} \cup \set{u_j''}_{j \in V} \; ,
\]
with source $\usource$ and target $\utarget$.
Henceforth, when we refer to the upper (lower) bound of an edge,
we mean the upper (lower) bound on the flow through it.
There are four kinds of edges:
\begin{enumerate}
\item
An edge $\usource \rightarrow u_\omega$ with upper and lower
bounds both equaling to $M'$.
\item
$u_\omega \rightarrow u'_i$ for every $i \in V$,
with the upper and lower bounds
$\lfloor \sum_{j \in V} \phifp_{i,j} \rfloor$ and
$\lceil \sum_{j \in V}  \phifp_{i,j} \rceil$, respectively.
\item
$u'_i \rightarrow u''_j$ for every $i, j \in V$,
with the upper and lower bounds
$\lfloor \phifp_{i,j} \rfloor$ and
$\lceil  \phifp_{i,j} \rceil$, respectively.
\item
$u''_j \rightarrow \utarget$ for every $j \in V$,
with the upper and lower bounds
$\lfloor \sum_{i \in V} \phifp_{i,j} \rfloor$ and
$\lceil \sum_{i \in V}  \phifp_{i,j} \rceil$, respectively.
\end{enumerate}

We claim that~\Equation{eq:preflow_sumij}--\Equation{eq:preflow_sumi}
can be satisfied if a legal integer flow exists:
simply take $\dalpha_{i,j}$ as
the flow on the edge from $u'_i$ to $u''_j$.

It is well known that if a legal \emph{real} flow exists for
a flow network with integer upper and lower bounds on the edges,
then a legal \emph{integer} flow exists
as well~\cite[Theorem 6.5]{Ahuja+:93}.
Moreover, such a flow can be
efficiently found~\cite[\S 6.7]{Ahuja+:93}.
To finish the proof, we now exhibit such a legal real flow:
\begin{enumerate}
\item
The flow on the edge $\usource \rightarrow u_\omega$ is
$\sum_{i,j \in V} \phifp_{i,j} = M'$.
\item
The flow on an edge $u_\omega \rightarrow u'_i$ is
$\sum_{j \in V} \phifp_{i,j}$.
\item
The flow on an edge $u'_i \rightarrow u''_j$ is $\phifp_{i,j}$.
\item
The flow on an edge $u''_j \rightarrow \utarget$ is
$\sum_{i \in V} \phifp_{i,j}$.
\end{enumerate}
\end{proof}

\begin{figure}[ht]
\centering
\myincludegraphics{flow}{5}

\[
\Phifp = \left(
\begin{array}{ccc}
3.03 & 2.54 & 0 \\
1.65 & 0 & 0.89 \\
0.89 & 0 & 0
\end{array}
\right) \; , \qquad
\Dalpha  = \left(
\begin{array}{ccc}
4 & 2 & 0 \\
2 & 0 & 1 \\
0 & 0 & 0
\end{array}
\right) \; .
\]

\caption{\textbf{Running Example (3):} The flow network derived from $\Phifp$ in
Running Example 2. An edge labeled $a;\mathbf{b}$ has lower and upper
bounds $\floor{a}$ and $\ceil{a}$, respectively. A legal real flow is given by
$a$. A legal integer flow is given by $\mathbf{b}$. The matrix $\Dalpha$
resulting from the legal integer flow is given, as well as the matrix $\Phifp$ (again).}
\label{fig:running3}
\end{figure}

For the remaining part of this section, we assume that
$\Dalpha$ is a good quantization of $\Phifp$
(say, $\Dalpha$ is computed by solving the integer flow
problem in the last proof). The next lemma states that $\Dalpha$ ``almost'' satisfies (\ref{eq:sum_up}).

\begin{lemm}
\label{lemm:disc} Let $\bfralpha = (\ralpha_i) = \bfone \cdot
\Dalpha^T$ and $\bfrbeta = (\rbeta_i) = \bfone \cdot \Dalpha$.
Then, for all $i \in V$,
\[
\ralpha_i - \rbeta_i \in \set{-1,0,1} \; .
\]
\end{lemm}
\begin{proof}
From \Equation{eq:preflow_sumj}, we get that for all $i \in V$,
\begin{equation}
\label{eq:ralphi_bounded}
\textstyle \lfloor \sum_{j \in V} \phifp_{i,j} \rfloor
\leq \ralpha_i 
\leq \lceil \sum_{j \in V} \phifp_{i,j} \rceil \; .
\end{equation}
Recall that \Equation{eq:sum_up} is satisfied if we replace $\Da$
by $\Phifp$. Thus, by \Equation{eq:preflow_sumi}, we have that
\Equation{eq:ralphi_bounded} also holds if we replace $\ralpha_i$ by $\rbeta_i$.
We conclude that $\size{\ralpha_i - \rbeta_i} \leq 1$. The proof follows from the fact that entries of $\Dalpha$ are integers, and thus so are those of $\bfralpha$ and $\bfrbeta$.
\end{proof}


The following lemma will be the basis for augmenting $\Dalpha$ so that (\ref{eq:sum_up}) is satisfied.
\begin{lemm}
\label{lemm:Fmat}
Fix two distinct vertices $s,t \in V$. We can efficiently find a $\size{V} \times \size{V}$ matrix $F^{(s,t)}=F=(f_{i,j})_{i,j \in V}$ with non-negative integer entries, such that the following three conditions hold.
\begin{enumerate}
\renewcommand{\theenumi}{\roman{enumi}}
\renewcommand{\labelenumi}{(\theenumi)}
\item
\label{eq:fsumEntries}
\[
\bfone \cdot F \cdot \bfone^T \leq \diam{G} \; .
\]
\item For all $i,j \in V$,
\label{eq:fsane}
\[
f_{i,j} > 0 \;\; \mbox{only if} \;\; a_{i,j} > 0 \; .
\]
\item
Denote $\bfxi =  \bfone \cdot F^T$ and $\bfx = \bfone \cdot F$. Then, for all $i \in V$,
\[
\label{eq:fsurp}
x_i- \xi_i = 
\begin{cases}
-1 & \mbox{if $i = s$,} \\
1 & \mbox{if $i = t$,} \\
0 &  \mbox{otherwise.}
\end{cases}
\]
\end{enumerate}
\end{lemm}

\begin{proof}
Let $k_1=s,k_2,k_3 \ldots, k_{\ell+1}=t$ be the vertices along a
shortest path from $s$ to $t$ in $G$. For all $i,j \in V$, define
\begin{equation}
\label{eq:f}
f_{i,j} = \size{\set{1 \leq h \leq \ell : k_h = i \;\; \mbox{and} \;\; k_{h+1} = j }} \; .
\end{equation}
Namely, $f_{i,j}$ is the number of edges from $i$ to $j$ along the path.

Conditions (\ref{eq:fsumEntries}) and (\ref{eq:fsane}) easily follow from (\ref{eq:f}). Condition (\ref{eq:fsurp}) follows from the fact that $\xi_i$ ($x_i$) is equal to the number of edges along the path for which $i$ is the start (end) vertex of the edge.
\end{proof}

The matrix $\Dalpha$ will be the basis for computing
a good multiplicity matrix $\Da$, as we demonstrate in
the proof of the next theorem.

\begin{theo}
\label{theo:DfromPtilde}
Let $\Dalpha = (\dalpha_{i,j})$ be a good quantization of $\Phifp$.
There exists a multiplicity matrix $\Da = (\da_{i,j})$
with respect to $G$ and $M$, such that
\begin{enumerate}
\item \label{it:daGreaterThanDalpha}
$\da_{i,j} \geq \dalpha_{i,j}$ for all $i, j \in V$, and---
\item
$\Malpha \leq \bfone \cdot \Da \cdot \bfone^T \leq M$
\end{enumerate}
(where $\Malpha$ is as defined in~\Equation{eq:Malpha}).
Moreover, the matrix $\Da$ can be found by an efficient algorithm.
\end{theo}

\begin{proof}
Consider a vertex $i \in V$. If $\rbeta_i > \ralpha_i$,
then we say that vertex $i$ has
a \emph{surplus} of $\rbeta_i - \ralpha_i$. In this case, by Lemma~\ref{lemm:disc}, we have that the surplus is equal to 1.
On the other hand, if $\rbeta_i < \ralpha_i$ then vertex $i$ has
a \emph{deficiency} of $\ralpha_i - \rbeta_i$, which again is equal to 1.

Of course, since
$\sum_{i \in V}  \ralpha_i = \sum_{i \in V} \rbeta_i = \Malpha$,
the total surplus is equal to the total deficiency,
and both are denoted by $\Surplus$:
\begin{equation}
\label{eq:surplus}
\Surplus = \sum_{i \in V}  \mmax{0, \rbeta_i {-}  \ralpha_i} =
- \sum_{i \in V} \mmin{0, \rbeta_i {-}  \ralpha_i} \; .
\end{equation}


Denote the vertices with surplus  as $(s_k)_{k=1}^\Surplus$ and the vertices with deficiency as $(t_k)_{k=1}^\Surplus$. Recalling the matrix $F$ from Lemma~\ref{lemm:Fmat}, we define
\[
\Da = \Dalpha + \sum_{k=1}^\Surplus F^{(s_k,t_k)} \; .
\]

We first show that $\Da$ is a valid multiplicity matrix. Note that $\Surplus \leq \size{V}/2$. Thus, (\ref{eq:sum_entries}) follows from (\ref{eq:Malpha}), (\ref{eq:preflow_sumij}), and (\ref{eq:fsumEntries}). The definitions of surplus and deficiency vertices along with (\ref{eq:fsurp}) give (\ref{eq:sum_up}). Lastly, recall that (\ref{eq:sane_entries}) is satisfied if we replace $\da_{i,j}$ by $\phifp_{i,j}$. Thus, by (\ref{eq:preflow_entries}), the same can be said for $\dalpha_{i,j}$. Combining this with (\ref{eq:fsane}) yields (\ref{eq:sane_entries}).

Since the entries of $F^{(s_k,t_k)}$ are non-negative for every $k$, we must have that $\da_{i,j} \geq \dalpha_{i,j}$ for all $i, j \in V$. This, together with (\ref{eq:sum_entries}) and (\ref{eq:preflow_sumij}), implies in turn that $\Malpha \leq \bfone \cdot \Da \cdot \bfone^T \leq M$.

\end{proof}

\begin{example}[Running Example (4):] For the matrix $\Dalpha$ in Figure~\ref{fig:running3}, we have
\[
\bfrbeta = \left(
\begin{array}{ccc}
6 & 2 & 1
\end{array}
\right) \; , \quad \bfralpha = \left(
\begin{array}{ccc}
6 & 3 & 0
\end{array}
\right) \; .
\]
Thus, $\Surplus = 1$. Namely, the vertex $\theta$ has a surplus while the vertex $\beta$ has a deficiency. Taking $s = \theta$ and $t = \beta$ we get
\[
F^{(s,t)} =
\left(
\begin{array}{ccc}
0 & 1 & 0 \\
0 & 0 & 0 \\
1 & 0 & 0
\end{array}
\right) \; , \;\; \mbox{and} \quad \Da =
\left(
\begin{array}{ccc}
4 & 3 & 0 \\
2 & 0 & 1 \\
1 & 0 & 0
\end{array}
\right) \; .
\]

\end{example}

Now that Theorem~\ref{theo:DfromPtilde} is proved, we are in a position to prove
our main result, Theorem~\ref{theo:main}. Essentially, the proof involves
using the Stirling approximation and taking into account the various
quantization errors introduced into $\Da$. The proof itself is given in the
Appendix.

\section{Enumerative coding into sequences with a given Markov type}
\label{sec:1D}
The main motivation for our methods is 2-D constrained coding. However, in this section, we show that they might be interesting in certain aspects of 1-D coding as well. Given a labeled graph $G$, a classic method for building an encoder for the 1-D constraint $S(G)$ is the state-splitting algorithm \cite{ACH:83}. The rate of an encoder built by \cite{ACH:83} approaches the capacity of $S(G)$. Also, the word the encoder outputs has a corresponding path in $G$, with the following favorable property: the probability of traversing a certain edge approaches the maxentropic probability of that edge (assuming an unbiased source distribution). However, what if we'd like to build an encoder with a different probability distribution on the edges? This scenario may occur, for example, when there is a requirement that all the output words of a given length $N$ that are generated by the encoder have a prescribed Hamming weight\footnote{We remark in passing that one may use convex programming techniques (see \cite[\S V]{MarcusRoth:92}) in order to efficiently solve the following optimization problem: find a probability distribution on the edges of $G$ yielding a stationary Markov chain with largest possible entropy, subject to a set of edges (such as the set of edges with label `1') having a prescribed cumulative probability.}.

More formally, suppose that we are given a labeled graph $G=(V,E,L)$; to make the exposition simpler, suppose that $G$ does not contain parallel edges. Let $Q$ and $\bfpi$ be a transition matrix and a stationary probability distribution corresponding to a stationary (but not necessarily maxentropic) Markov chain $\Markovchain$ on $G$. We assume w.l.o.g.\ that each edge in $G$ has a positive conditional probability. We are also given an integer $M$, which we will shortly elaborate on.

We first describe our encoder in broad terms, so as that its merits will be obvious. Let $D$ and $N$ be as previously defined, and let $R_T(D)$ be specified shortly.  We start at some fixed vertex $v_0 \in V$. Given $M \cdot R_T(D)$ information bits, we traverse a soon to be defined cyclic path of length $N$ in $G$. The concatenation of the edge labels along the path is the word we output. Of course, since the path is cyclic, the concatenation of such words is indeed in $S(G)$. Moreover, the path will have the following key property: the number of times an edge from $i$ to $j$ is traversed equals $d_{i,j}$. Namely, if we uniformly pick one of the $N$ edges of the path, the probability of picking a certain edge $e$ is constant (not a function of the input bits), and is equal to the probability of traversing $e$ on the Markov chain $\Markovchain$, up to a small quantization error. The rate $R_T$ of our encoder will satisfy (\ref{eq:Rbound}), where we replace $R$ by $R_T$ and $\capacity(S)$ by the entropy of $\Markovchain$. We would like to be able to exactly specify the path length $N$ as a design parameter. However, we specify $M$ and get an $N$ between $M$ and $M - \floor{ \size{V} \diam{G}/2 }$.

Our encoding process will make use of an \emph{oriented tree}, a term which we will now define.
A set of edges $T \subseteq E$ is an oriented tree of $G$ with root $v_0$ if $|T| = |V|-1$ and for each $u \in V$ there exists a path from $u$ to $v_0$ consisting entirely of edges in $T$ (see Figure~\ref{fig:orientedTree}). Note that if we reverse the edge directions of an oriented tree, we get a directed tree as defined in \cite[Theorem 2.5]{Even:79}. Since reversing the directions of all edges in an irreducible graph results in an irreducible graph, we have by \cite[Lemma 3.3]{Even:79} that an oriented tree $T$ indeed exists in $G$, and can be efficiently found. So, let us fix some oriented tree $T$ with root $v_0$. By \cite[Theorem 2.5]{Even:79}, we have that every vertex $u \in V$ which is not the root $v_0$ has an out-degree equal to 1. Thus, for each such vertex $u$ we may define $\parent(u)$ as the destination of the single edge in $T$ going out of $u$.

\begin{figure}[ht]
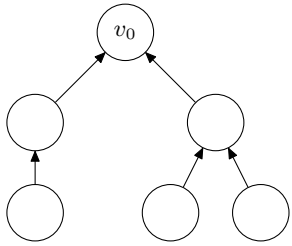

\centering
\myincludegraphics{flow}{9}
\caption{Oriented tree with root $v_0$.}
\label{fig:orientedTree}
\end{figure}

We now elaborate on the encoding process. The encoding consists of two steps. In the first step, we map the information bits to a collection of lists. In the second step, we use the lists in order to define a cyclic path.

First step: Given $M \cdot R_T(D)$ information bits, we build for each vertex $i \in V$ a list $\bflambda^{(i)}$ of length $r_i$,
\[
\bflambda^{(i)} = (\lambda^{(i)}_1,\lambda^{(i)}_2,\ldots,\lambda^{(i)}_{r_i}) \; .
\]
The entries of each $\bflambda^{(i)}$ are vertices in $V$. Moreover, the following properties are satisfied for all $i$: 
\begin{itemize}
\item The number of times $j$ is an entry in $\bflambda^{(i)}$ is exactly $d_{i,j}$.
\item If $i \neq v_0$, then the last entry of the list equals the parent of $i$. Namely,
\[
\lambda^{(i)}_{r_i} = \parent(i) \; .
\]
\end{itemize}

Recalling (\ref{eq:R_D}), a simple calculation shows that the number of possible list collections is
\begin{equation}
\label{eq:DeltaT}
\Delta_T = \Delta \cdot \prod_{i \in V \setminus \set{v_0}} \frac{d_{i,\parent(i)}}{r_i} \; .
\end{equation}
Thus, we define the rate of encoding as
\[
R_T = \frac{\lfloor \log_2 \Delta_T \rfloor}{M} \; .
\]
Also, note that as in the 2-D case, we may use enumerative coding in order to efficiently map information bits to lists.

Second step: We now use the lists $\bflambda^{(i)}$, $i \in V$,  in order to construct a cyclic path starting at vertex $v_0$. We start the path at $v_0$ and build a length-$N$ path according to the following rule: when exiting vertex $i$ for the $k$th time, traverse the edge going into vertex $\lambda^{(i)}_k$.

Of course, our encoding method is valid (and invertible) iff we may always abide by the above-mentioned rule. Namely, we don't get ``stuck'', and manage to complete a cyclic path of length $N$. This is indeed the case: define an auxiliary graph $G(D)$ with the same vertex set, $V$, as $G$ and $d_{i,j}$ parallel edges from $i$ to $j$ (for all $i,j \in V$). First, recall that for sufficiently large $M$, the presence of an edge from $i$ to $j$ in $G$ implies that $d_{i,j} > 0$. Thus, since $G$ was assumed to be irreducible, $G(D)$ is irreducible as well. Also, an edge in $T$ from $i$ to $j$ implies the existence of an edge in $G(D)$ from $i$ to $j$. Secondly, note that by (\ref{eq:sum_up}), the number of times we are supposed to exit a vertex is equal to the number of times we are supposed to enter it.  The rest of the proof follows from \cite[p. 56, Claim 2]{Stanley:99}, applied to the auxiliary graph $G(D)$. Namely, our encoder follows directly from van Aardenne-Ehrenfest and de Bruijn's \cite{vAEdB:51} theorem on counting Eulerian cycles in a graph. 

We now return to the rate, $R_T$, of our encoder. From (\ref{eq:Malpha}), (\ref{eq:preflow_entries}), (\ref{eq:preflow_sumi}) and Theorem \ref{theo:DfromPtilde}, we see that for $M$ sufficiently large, $\Delta_T$ is greater than some positive constant times $\Delta$. Thus, (\ref{eq:Rbound}) still holds if we replace $R$ by $R_T$ and $\capacity(S)$ by the entropy of $\Markovchain$.

\section{An example, and two improvement techniques}
\label{sec:twoOpt}
Recall from Section~\ref{sec:twoDimensionalConstraints} the square
constraint: its elements are all the binary arrays in which no two `1' symbols are adjacent on a row, column, or diagonal. By employing the methods presented in \cite{CalkinWilf:97}, we may calculate an upper bound on the rate of the constraint. This turns out to be $0.425078$. We will show
an encoding/decoding method with rate slightly larger
than $0.396$ (about $93\%$ of the upper bound). In order to do this, we assume
that the array has $100{,}000$ columns. Our encoding method has a fixed rate and
has a vertical window of size 2 and vertical anticipation 0.

We should point out now that a straightforward implementation of the
methods we have previously defined gives a rate which is strictly
$\emph{less}$ than $0.396$. Namely, this section also
outlines two improvement techniques which help boost the rate.

We start out as in the example given in
Section~\ref{sec:twoDimensionalConstraints}, except that the width of the data strips is now
$\wt=9$ (the width
of the merging strips remains $\wm=1$). The graph $G$ we choose produces all width-$\wt$ arrays satisfying the square constraint, and we take the merging strips to be
all-zero. Our array has $100{,}000$ columns, so we have $M = 10{,}000$ tracks
(the last, say, column of the array will essentially be unused; we can set all
of its values to 0).

Define the normalized capacity as
\[
\frac{\capacity(S(G))}{\wt + \wm} \; .
\]
The graph $G$ has $|V| = 89$ vertices and normalized capacity
\[
\frac{\capacity(S(G))}{\wt + \wm} \approx \frac{\capacity(S(G))}{\wt + \wm(1 - 1/M)} \approx  0.402 \; .
\]
This number is about $94.5\%$ from the upper bound on the capacity of our 2-D
constraint. Thus, as expected, there is an inherent loss in choosing to
model the 2-D constraint as an essentially 1-D constraint. Of course, this loss
can be made smaller by increasing $\wt$ (but the graph $G$ will grow as well).

From Theorem~\ref{theo:main}, the rate of our encoder will
approach the normalized capacity of $0.402$ as the number of tracks $M$
grows. So, once the graph $G$ is chosen, the parameter we should be comparing
ourselves to is the normalized capacity. We now apply the methods defined
in Section~\ref{sec:computingmultiplicitymatrix} and find a multiplicity matrix
$\Da$. Recall that the matrix $\Da$ defines an encoder. In our case, this
encoder has a rate of about $0.381$. This is $94\%$ of the normalized capacity, and is quite  disappointing (but the improvements shown in Sections~\ref{ssec:Moore} and \ref{ssec:breakmerge} below are going to improve this rate). On
the other hand, note that if we had limited ourselves to encode to each track
independently of the others, then the best rate we could have hoped for with 0
vertical anticipation turns out to be $0.3$ (see \cite[Theorem
5]{MarcusRoth:91}).

\subsection{Moore-style reduction}
\label{ssec:Moore} We now define a graph $\mathsf{G}$ which we call the reduction of $G$. Essentially, we will encode by constructing paths in $\mathsf{G}$, and then translate these to paths in $G$. In both $G$ and $\mathsf{G}$, the maxentropic distributions have the same entropy. The main virtue of $\mathsf{G}$ is that it often has less vertices and edges compared to $G$. Thus, the penalty in (\ref{eq:Rbound}) resulting from using a finite number of tracks will often be smaller.

For $s \geq 0$, we now recursively define the concept of $s$-equivalence (very much like in the Moore algorithm \cite[page 1660]{MarcusRothSiegel:98}).
\begin{itemize}
\item For $s=0$, any two vertices $v_1,v_2 \in V$ are 0-equivalent.
\item For $s > 0$, two vertices $v_1,v_2 \in V$ are $s$-equivalent iff 1) the two vertices $v_1,v_2$ are $(s-1)$-equivalent, and 2) for each $(s-1)$-equivalence class $\sfc$, the number of edges from $v_1$ to vertices in $\sfc$ is equal to the number of edges from $v_2$ to vertices in $\sfc$.
\end{itemize}
Denote by $\Pi_s$ the partition induced by $s$-equivalence. For the graph $G$ given in Figure~\ref{fig:example:G},
\begin{eqnarray*}
&&\Pi_0 = {\scriptstyle \set{0000,0001,0010,0100,0101,1000,1001,1010}} \; ,\\
&&\Pi_{s \geq 1} = \scriptstyle \set{0000},\set{0010,0100},\set{1000,0001},\set{1010,1001,0101} \; .
\end{eqnarray*}
Note that, by definition, $\Pi_{s+1}$ is a refinement of $\Pi_s$. Thus, let $s'$ be the smallest $s$ for which $\Pi_s$ = $\Pi_{s+1}$. The set $\Pi_{s'}$ can be efficiently found (essentially, by the Moore algorithm \cite[page 1660]{MarcusRothSiegel:98}).

Define a (non-labeled) graph $\mathsf{G}=(\mathsf{V},\mathsf{E})$ as follows. The vertex set of $\mathsf{G}$ is 
\[
\mathsf{V}=\Pi_{s'} \; .
\]
For each $\sfc \in \mathsf{V}$, let $v(\sfc)$ be a fixed element of $\sfc$ (if $\sfc$ contains more than one vertex, then pick one arbitrarily). Also, for each $v \in V$, let $\sfc(v)$ be the class $\sfc \in \mathsf{V}$ such that $v \in \sfc$. Let $\sigma_G(e)$ ($\sigma_\sfG(e)$) and $\tau_G(e)$ ($\tau_\sfG (e)$) denote the start and end vertex of an edge $e$ in $G$ ($\sfG$), respectively. The edge set $\mathsf{E}$ is defined as
\begin{equation}
\label{eq:Etag}
\mathsf{E} = \bigcup_{\sfc \in \mathsf{V}} \set{e \in E : \sigma_G(e) = v(\sfc)} \; ,
\end{equation}
where
\[
\sigma_{\sfG}(e) = \sfc(\sigma_G(e)) \quad \mbox{and} \quad \tau_{\sfG}(e) = \sfc(\tau_G(e)) \; .
\]
Namely, the number of edges from $\sfc_1$ to $\sfc_2$ in $\sfG$ is equal to the number of edges in $G$ from some fixed $v_1 \in \sfc_1$ to elements of $\sfc_2$, and, by the definition of $s'$, this number does not depend on the choice of $v_1$. The graph $\sfG$ is termed the \emph{reduction} of $G$. The reduction of $G$ from Figure~\ref{fig:example:G} is given in Figure~\ref{fig:example:Greduced}. Note that since $G$ was assumed to be irreducible, we must have that $\sfG$ is irreducible as well.

\begin{figure}[ht]
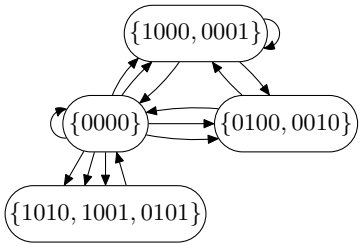

\centering
\myincludegraphics{flow}{7}
\caption{Reduction of the graph $G$ from Figure~\ref{fig:example:G}.}
\label{fig:example:Greduced}
\end{figure}

\begin{lemm}
The entropies of the maxentropic Markov chains on $G$ and $\sfG$ are equal.
\end{lemm}

\begin{proof}
Let $\mathsf{A}=\mathsf{A}_{\sfG}$ be the adjacency matrix of $\sfG$, and recall that $A = A_G$ is the adjacency matrix of $G$. Let $\mathsf{\lambda}'$ and $\bfx' = (x'_\sfc)_{\sfc \in \mathsf{V}}$ be the Perron eigenvalue and right Perron eigenvector of $\mathsf{A}$, respectively \cite[\S 3.1]{MarcusRothSiegel:98}. Next, define the vector $\bfx=(x_v)_{v \in V}$ as
\[
x_v = x'_{\sfc(v)} \; .
\]
It is easily verifiable that $\bfx$ is a right eigenvector of $A$, with eigenvalue $\lambda'$. Now, since $\bfx'$ is a Perron eigenvector of an irreducible matrix, each entry of it is positive. Thus, each entry of $\bfx$ is positive as well. Since $A$ is irreducible, we must have that $\bfx$ is a Perron eigenvector of $A$. So, the Perron eigenvalue of $A$ is also $\lambda'$.
\end{proof}

The next lemma essentially states that we can think of paths in $\sfG$ as if they were paths in $G$.
\begin{lemm}
\label{lemm:pathConversion}
Let $\ell \geq 1$. Fix some $\sfc_0,\sfc_{\ell+1} \in \mathsf{V}$, and $v_0 \in \sfc_0$.  There exists a one-to-one correspondence between the following sets. First set: paths of length $\ell$ in $\sfG$ with start vertex $\sfc_0$ and end vertex $\sfc_{\ell+1}$. Second set: paths of length $\ell$ in $G$ with start vertex $v_0$ and end vertex in $\sfc_{\ell+1}$. 

Moreover, for $1 \leq t \leq \ell-1$, the first $t$ edges in a path belonging to the second set are a function of only the first $t$ edges in the respective path in the first set.
\end{lemm}
\begin{proof}
We prove this by induction on $\ell$. For $\ell = 1$, we have
\begin{eqnarray*}
&& \size{\set{e \in \mathsf{E} : \sigma_{\sfG}(e) = \sfc_0 \; , \quad \tau_{\sfG}(e) = \sfc_1}} = \\
&& \size{\set{e \in E : \sigma_{G}(e) = v_0 \; , \quad \tau_{G}(e) \in \sfc_1}} \; .
\end{eqnarray*}
To see this, note that we can assume w.l.o.g.\ that $v_0 = v(\sfc_0)$, and then recall (\ref{eq:Etag}). For $\ell > 1$, combine the claim for $\ell - 1$ with that for $\ell=1$.
\end{proof}

Notice that $\diam{\sfG} \leq \diam{G}$. We now show why $\sfG$ is useful.
\begin{theo}
Let $\mathsf{D}$ be the multiplicity matrix found by the methods previously outlined, where we replace $G$ by $\sfG$. Let $\sfN = \bfone \cdot \mathsf{D} \cdot \bfone^T$. We may efficiently encode (and decode) information to $\GKK$ in a row-by-row manner at rate $R(\mathsf{D})$.
\end{theo}

\begin{proof}
We conceptually break our encoding scheme into two steps. In the first step, we ``encode'' (map) the information into $\sfN$ paths in $\sfG$, each path having length $\ell$. We do this as previously outlined (through typical vertices and edges in $\sfG$). Note that this step is done at a rate of $R(\mathsf{D})$. In the second step, we map each such path in $\sfG$ to a corresponding path in $G$. By Lemma~\ref{lemm:pathConversion}, we can indeed do this (take $\sfc_0$ as the first vertex in the path, $\sfc_{\ell+1}$ as the last vertex, and $v_0 = v(\sfc_0)$).

By Lemma~\ref{lemm:pathConversion} we see that this two-step encoding scheme can easily be modified into one that is row-by-row.
\end{proof}
Applying the reduction to our running example (square constraint with $\wt=9$ and $\wm=1$), reduces the number of vertices from $89$ in $G$ to $34$ in $\sfG$. The computed $\mathsf{D}$ increases the rate to about $0.392$, which is  $97.5\%$ of the normalized capacity.

\subsection{Break-merge}
\label{ssec:breakmerge}
Let $\sfGKK$ be the $\sfN$th Kronecker power of the Moore-style reduction $\sfG$. Recall that the rate of our encoder is 
\[
R(\mathsf{D}) = \frac{\floor{\log_2 \Delta}}{M} \; ,
\]
where $\Delta$ is the number of typical edges in $\sfGKK$ going out of a typical vertex. The second improvement involves expanding the definition of a typical edge, thus increasing $\Delta$. This is best explained through an example. Suppose that $\sfG$ has Figure~\ref{fig:breakMerge} as a subgraph; namely, we show all edges going out of vertices $\alpha$ and $\beta$. Also, let the numbers next to the edges be equal to the corresponding entries in $\mathsf{D}$. The main thing to notice at this point is that if the edges to $\epsilon$ and $\zeta$ are deleted (``break''), then $\alpha$ and $\beta$ have exactly the same number of edges from them to vertex $j$, for all $j\in V$ (after the deletion of edges, vertices $\alpha$ and $\beta$ can be ``merged'').
\begin{figure}[ht]
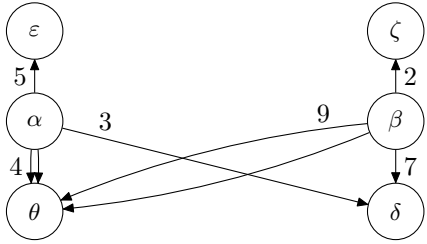

\centering
\myincludegraphics{flow}{8}
\caption{Break-merge example graph.}
\label{fig:breakMerge}
\end{figure}

Let $\bfv$ be a typical vertex. A short calculation shows that the number of entries in $\bfv$ that are equal to $\alpha$ ($\beta$) is $5+4+3 = 12$ ($9+7+2=18$). Recall that the standard encoding process consists of choosing a typical edge $\bfe$ going out of the typical vertex $\bfv$ and into another typical vertex $\bfv'$. We now briefly review this process. Consider the $12$ entries in $\bfv$ that are equal to $\alpha$. The encoding process with respect to them will be as follows (see Figure \ref{fig:breakMerge1}):
\begin{itemize}
\item Out of these $12$ entries, choose $5$ for which the corresponding entry in $\bfv'$ will be $\varepsilon$. Since there is exactly one edge from $\alpha$ the $\varepsilon$ in $\sfG$, the corresponding entries in $\bfe$ must be equal to that edge.
\item Next, from the remaining $7$ entries, choose $4$ for which
the corresponding entries in $\bfv'$ will be $\theta$. There are two parallel edges from $\alpha$ to $\theta$, so choose which one to use in the corresponding entries in $\bfe$.
\item We are left with $3$ entries, the corresponding entries in $\bfv'$ will be $\delta$. Also, we have one option as to the corresponding entries in $\bfe$.
\end{itemize}
A similar process is applied to the entries in $\bfv$ that are equal to $\beta$. Thus, the total number of options with respect to these entries is
\[
\frac{12! \cdot 2^4}{5! \cdot 4! \cdot 3!} \cdot \frac{18! \cdot 2^9}{2! \cdot 9! \cdot 7!} \approx 3.97 \cdot 10^{14} \; .
\]

\begin{figure}[ht]
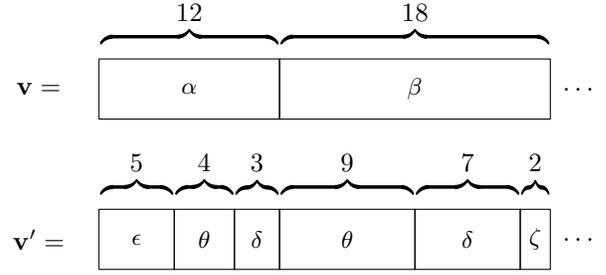

\centering
\myincludegraphics{flow}{10}
\caption{Illustration of the entries in two typical vertices $\bfv$, $\bfv'$, where we got from $\bfv$ to $\bfv'$ by the standard encoding process.}
\label{fig:breakMerge1}
\end{figure}

Next, consider a different encoding process (see Figure \ref{fig:breakMerge2}).
\begin{itemize}
\item Out of the $12$ entries in $\bfv$ that are equal to $\alpha$, choose $5$ for which the corresponding entry in $\bfv'$ will be $\varepsilon$. As before, the corresponding entries in $\bfe$ have only one option.
\item Out of the $18$ entries in $\bfv$ that are equal to $\beta$, choose $2$ for the corresponding entry in $\bfv'$ will be $\zeta$. Again, one option for entries in $\bfe$.
\item Now, of the remaining $23$ entries in $\bfv$ that are equal to $\alpha$ or $\beta$, choose $4+9=13$ for which the corresponding entry in $\bfv'$ will be $\theta$. We have two options for the entries in $\bfe$.
\item We are left with $3+7=10$ entries in $\bfv$ that are equal to $\alpha$ or $\beta$. These will have $\delta$ as the corresponding entry in $\bfv'$, and one option in $\bfe$.
\end{itemize}

\begin{figure}[ht]
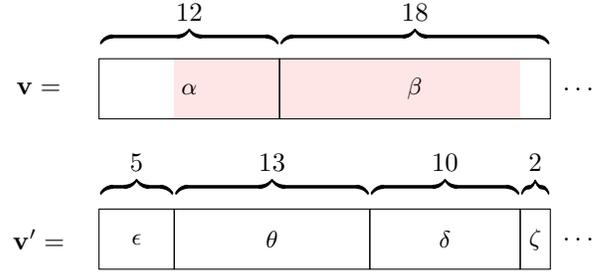

\centering
\myincludegraphics{flow}{11}
\caption{Illustration of the entries in two typical vertices $\bfv$, $\bfv'$, where we got from $\bfv$ to $\bfv'$ by the improved encoding process. The shaded part corresponds to vertices that were merged.}
\label{fig:breakMerge2}
\end{figure}

Thus, the total number of options is now
\[
\binom{12}{5} \cdot \binom{18}{2} \cdot \frac{23! \cdot 2^{13}}{13! \cdot 10!} \approx 1.14 \cdot 10^{15} \; .
\]
The important thing to notice is that in both cases, we arrive at a typical vertex $\bfv'$.

To recap, we first ``broke'' the entries in $\bfv$ that are equal to $\alpha$ into two groups: Those which will have $\varepsilon$ as the corresponding entry in $\bfv'$ and those which will have $\theta$ or $\delta$ as the corresponding entry. Similarly, we broke entries in $\bfv$ that are equal to $\beta$ into two groups. Next, we noticed that of these four groups, two could be ``merged'', since they were essentially the same. Namely, removing some edges from the corresponding vertices in $\sfG$ resulted in vertices which were mergeable.

Of course, these operations can be repeated. The hidden assumption is that the sequence of breaking and merging is fixed, and known to both the encoder and decoder. The optimal sequence of breaking and merging is not known to us. We used a heuristic. Namely, choose two vertices such that the sets of edges emanating from both have a large overlap. Then, break and merge accordingly. This was done until no breaking or merging was possible. We got a rate of  about $0.396$, which is  $98.5\%$ of the normalized capacity.

\section{Fast enumerative coding}
\label{sec:fastenumerativecoding}

Recall from Section~\ref{sec:encoder} that in
the course of our encoding algorithm, we make use of a procedure which encodes information
into fixed-length binary words of constant weight.
A way to do this would be to use enumerative coding~\cite{Cover:73}.
Immink~\cite{Immink:97} showed a method to significantly improve
the running time of an instance of enumerative coding, with a typically
negligible penalty in terms of rate.
We now briefly show how to tailor Immink's method to our needs.

Denote by $n$ and $\delta$ the length and Hamming weight, respectively,
of the binary word we encode into. Some of our variables will be \emph{floating-point} numbers with a mantissa of
$\mu$ bits and an exponent of $\epsilon$ bits:
each floating-point number is of the form
$\float{x} = a \cdot 2^b$ where $a$ and $b$ are integers such that
\[
2^{\mu} \leq a < 2^{\mu+1} \quad \mbox{and} \quad -2^{\epsilon-1} \leq b < 2^{\epsilon-1} \; .
\]
Note that $\mu+\epsilon$ bits are needed to store such a number. Also, note that every positive real $x$ such that
\[
2^\mu \cdot 2^{-2^{\epsilon-1}} \leq x \leq (2^{\mu+1} - 1) \cdot 2^{2^{\epsilon-1}-1}
\]
has a floating point approximation $\float{x}$ with relative precision
\begin{equation}
\label{eq:fp}
\left( 1 - \frac{1}{2^{\mu}} \right) \leq \frac{\float{x}}{x} \leq \left( 1 + \frac{1}{2^{\mu}}  \right) \; .
\end{equation}
We assume the presence of two look-up tables. The first will contain
the floating-point approximations of $1!,2!,\ldots,n!$.
The second will contain the floating-point approximations of
$f(0), f(1), \ldots, f(\delta)$, where
\[
f(\chi) = f_\mu(\chi) = 1 - \frac{32 \chi+ 16}{2^{\mu}} \; .
\]
In order to exclude uninteresting cases, assume that $\mu \geq 10$ and is such that $f(\delta) \geq 1/2$. Also, take $\epsilon$ large enough so that $n!$ is less than the maximum number we can represent by floating point. Thus, we can assume that $\mu = O(\log \delta)$ and $\epsilon = O(\log n)$.

Notice that in our case,
we can bound both $n$ and $\delta$ from above by
the number of tracks $M$.
Thus, we will actually build beforehand two look-up tables of size
$2M(\mu+\epsilon)$ bits.

Let $\float{x}$ denote the floating-point approximation of $x$,
and let $\fmult$ and $\div$ denote floating-point multiplication and
division, respectively. For $0 \leq \chi \leq \kappa \leq n$ we define
\begin{equation*}
\ifbinom{\kappa}{\chi} =
\ceiling{(\float{\kappa!}\fmult
\float{f(\chi)})\div (\float{\chi!} \fmult \float{(\kappa-\chi)!})} \; .
\end{equation*}
Note that since we have stored the relevant numbers in our
look-up table, the time needed to calculate the above function is only
$O(\mu^2+\epsilon)$. The encoding procedure is given in
Figure~\ref{fig:enumerativeCoder}. We note the following points: 
\begin{itemize}
\item The variables $n$, $\psi$, $\delta$ and $\iota$ are integers (as opposed to floating-point numbers). \item In the subtraction of $\ifbinom{n-\iota}{\delta-1}$ from $\psi$ in line 5, the floating-point number $\ifbinom{n-\iota}{\delta-1}$ is ``promoted'' to an integer (the result is an integer).
\end{itemize}

\begin{figure}[hbt]
\centering
\begin{algorithm}
\textbf{Name:} $\textrm{EnumEncode}(n,\delta,\psi)$ \\
\vspace{-1ex}

\textbf{Input:}
Integers $n,\delta,\psi$ such that
$0 \leq \delta \leq n$ and $0 \leq \psi < \ifbinom{n}{\delta}$.

\textbf{Output:}
A binary word of length $n$ and weight $\delta$.
\vspace{2ex}

if ($\delta$ == $0$) // stopping condition:  \hfill /* 1 */ \\
\hspace*{1em} return $\underbrace{00 \ldots 0}_n$; \hfill /* 2 */ \\
for ($\iota \leftarrow 1$; $\iota \leq n-\delta+1$; $\iota$++) \{ \hfill /* 3 */ \\
\hspace*{1em} if ($\psi \geq \ifbinom{n-\iota}{\delta-1}$) \hfill /* 4 */ \\
\hspace*{2em} $\psi \leftarrow \psi - \ifbinom{n-\iota}{\delta-1}$; \hfill /* 5 */ \\
\hspace*{1em} else \hfill /* 6 */ \\
\hspace*{2em} return
$\underbrace{00..0}_{\iota-1} 1 \|
\textrm{EnumEncode}(n-\iota, \delta-1, \psi)$; \hfill /* 7 */ \\
\} \hfill /* 8 */ \\
\end{algorithm}
\caption{Enumerative encoding procedure for constant-weight
binary words.}
\label{fig:enumerativeCoder}
\end{figure}

We must now show that the procedure is valid, namely, that given a valid input, we produce a valid output. For our procedure, this reduce to showing two things: 1) If the stopping condition is not met, a recursive call will be made. 2) The recursive call is given valid parameters as well. Namely, in the recursive call, $\psi$ is non-negative. Also, for the encoding to be invertible, we must further require that 3) $\ifbinom{n}{0}=1$ for $n \geq 0$.

Condition 2 is clearly met, because of the check in line 4. Denote
\[
\fbinom{\kappa}{\chi} = (\float{\kappa!}\fmult
\float{f(\chi)})\div (\float{\chi!} \fmult \float{(\kappa-\chi)!}) 
\]
(and so, $\ifbinom{\kappa}{\chi} = \lceil\fbinom{\kappa}{\chi}\rceil$). Condition 3 follows from the next lemma.

\begin{lemm}
\label{lemm:boundRoundoff}
Fix $0 \leq \delta \leq n$. Then,
\[
\binom{n}{\delta} \cdot \left( 1 - \frac{32 (\delta+1)}{2^\mu} \right) \leq \fbinom{n}{\delta} \leq \binom{n}{\delta} \cdot \left( 1 - \frac{32 \delta}{2^\mu} \right) \; .
\]
\end{lemm}
\begin{proof}
The proof is essentially repeated invocations of (\ref{eq:fp}) on the various stages of computation. We leave the details to the reader.
\end{proof}

Finally, Condition 1 follows easily from the next lemma.
\begin{lemm}
Fix $0 \leq \delta \leq n$. Then,
\[
\ifbinom{n}{\delta} \leq \sum_{\iota = 1}^{n-\delta+1} \ifbinom{n-\iota}{\delta-1} \; .
\]
\end{lemm}

\begin{proof}
The claim will follow if we show that
\[
\fbinom{n}{\delta} \leq \sum_{\iota = 1}^{n-\delta+1} \fbinom{n-\iota}{\delta-1} \; .
\]
This is immediate from Lemma~\ref{lemm:boundRoundoff} and the binomial identity
\[
\binom{n}{\delta} = \sum_{\iota = 1}^{n-\delta+1} \binom{n-\iota}{\delta-1} \; .
\]
\end{proof}

Note that the penalty in terms of rate one suffers because of using our procedure (instead of plain enumerative coding) is negligible. Namely, $\log_2 \ifbinom{n}{\delta}$ can be made arbitrarily close to $\log_2 \binom{n}{\delta}$. Since we take $\epsilon = O(\log n)$ and $\mu = O(\log \delta)$, we can show by amortized analysis that the running time of the procedure is $O(n \log^2 n)$. Specifically, see \cite[Section 17.3]{CLRS:01}, and take the potential of the binary vector corresponding to $\psi$ as the number of entries in it that are equal to `$0$'. The decoding procedure is a straightforward ``reversal'' of the encoding procedure, and its running time is also $O(n \log^2 n)$.

\section{Appendix}
\label{sec:appendix}
\noindent\hspace{2em}{\itshape Proof of
Theorem~\ref{theo:main}:}
Let $\tilde{\Delta}$ be as in~\Equation{eq:R_D}, where we replace
$\da_{i,j}$ by $\dalpha_{i,j}$ and $\ra_i$ by $\ralpha_i$. By the
combinatorial interpretation of~\Equation{eq:R_D}, and the fact that
$\da_{i,j} \geq \dalpha_{i,j}$ for all $i, j \in V$, it easily follows that
$\Delta \geq \tilde{\Delta}$. Thus,
\[
R(\Da) \geq \frac{\lfloor \log_2 \tilde{\Delta} \rfloor}{M} =
\frac{\Malpha}{M} \cdot \frac{\lfloor \log_2 \tilde{\Delta}
\rfloor}{\Malpha} \; .
\]

Denote by $\sfe$ the base of natural logarithms. By Stirling's formula we have
\[
\log_2(t!) = t \log_2(t/\sfe) + O(\log t) \; ,
\]
and from \Equation{eq:R_D} we get that
\begin{multline*}
\log_2 \tilde{\Delta} = \sum_{i \in V} \ralpha_i \log_2(\ralpha_i/\sfe) - \sum_{i,j \in V}
\dalpha_{i,j} \log_2(\dalpha_{i,j}/\sfe)
\\
+ \sum_{i,j \in V} \dalpha_{i,j} \log_2(a_{i,j}) -
O(|V|^2 \log M) \; .
\end{multline*}
By (\ref{eq:preflow_sumij}) and \Equation{eq:preflow_entries},
\begin{multline*}
\sum_{i,j \in V} \dalpha_{i,j} \log_2(a_{i,j}) = 
\\
\sum_{i,j \in V} \phifp_{i,j} \log_2(a_{i,j}) - O\left(|V|^2
\log_2(\amax/\amin)\right) \; .
\end{multline*}
Since $\sum_{j} \dalpha_{i,j} = \ralpha_i$, we have
\begin{multline*}
\sum_{i \in V} \ralpha_i \log_2(\ralpha_i/\sfe) - \sum_{i,j \in V}
\dalpha_{i,j} \log_2(\dalpha_{i,j}/\sfe)
\\
= \sum_{i \in V} \ralpha_i \log_2(\ralpha_i) - \sum_{i,j \in V}
\dalpha_{i,j} \log_2(\dalpha_{i,j}) \; .
\end{multline*}
Moreover, by \Equation{eq:preflow_sumj} and \Equation{eq:preflow_entries}, the
RHS of the last equation equals
\[
\sum_{i \in V} \rfp_i \log_2(\rfp_i) - \sum_{i,j \in V}
\phifp_{i,j} \log_2(\phifp_{i,j}) - O(|V|^2) \; .
\]

We conclude that
\begin{multline*}
\log_2 \tilde{\Delta} = \sum_{i \in V} \rfp_i \log_2(\rfp_i) - \sum_{i,j \in
V}
\phifp_{i,j} \log_2(\phifp_{i,j})
\\
+ \sum_{i,j \in V} \phifp_{i,j} \log_2(a_{i,j}) - O\left(|V|^2 (\log
M \cdot \amax/\amin)\right) \; .
\end{multline*}

Lastly, recall that $\rfp_i = M' \pi_i$ and $\phifp_{i,j} = \rfp_i q_{i,j}$. Thus,
\[
\log_2 \tilde{\Delta} = M' H(\Markovchain) - O\left(|V|^2 (\log
M \cdot \amax/\amin)\right) \; ,
\]
where $H(\Markovchain)$ is the entropy of the stationary Markov chain $\Markovchain$  with transition matrix $Q$. Recall that $\Markovchain$ was selected to be maxentropic: $H(\Markovchain) = \capacity(S(G))$. This fact, along with \Equation{eq:Malpha} and a short calculation, finishes the proof.
\endproof

\section{Acknowledgments}
The first author would like to thank Roee Engelberg for very helpful
discussions.

\end{document}